\documentclass{llncs}
\usepackage[utf8]{inputenc}

\usepackage[english]{babel} 
\usepackage{amsmath} 
\usepackage{amsfonts} 
\usepackage{amssymb} 
\usepackage{mathdots}    
\usepackage{url} 
\usepackage{semantic}
\usepackage{cmll}
\usepackage{subfigure}
\usepackage{enumerate}   
\usepackage{bussproofs}
\usepackage{tikz}
\usepackage[final, backref, hyperindex, bookmarks]{hyperref}
\usepackage{stmaryrd}

\newcommand{\sub}[2]{#1[#2]}
\newcommand{\limp}{\multimap} 
\newcommand{\bang}{\oc} 

\newcommand{\hyp}[3]{#1:(#2, #3)}
\newcommand{\tertype}{{\sf Unit}}

\newcommand{\set}[1]{\{#1\}}
\newcommand{\st}[2]{{\sf set}(#1,#2)}

\newcommand{\rgtype}[2]{{\it {\sf Reg}_{#1} #2}}
\newcommand{\get}[1]{{\sf get}(#1)}

\newcommand{\store}[2]{#1 \Leftarrow #2}
\newcommand{\regtype}[2]{{\sf Reg}_{#1} #2}

\newcommand{\letd}[3]{\mathsf{let}\;\dagger\! #1 = #2\;\mathsf{in}\;#3}
\newcommand{\letb}[3]{\mathsf{let}\;\oc #1 = #2\;\mathsf{in}\;#3}

\newcommand{\letp}[3]{\mathsf{let}\;\mathsection #1 = #2\;\mathsf{in}\;#3}

\newcommand{\lets}[3]{\mathsf{let}\;#1 = #2\;\mathsf{in}\;#3}

\newcommand{\nat}{\mathsf{N}}

\newcommand{\reduc}{\rightarrow}
\newcommand{\treduc}{\stackrel{*}{\rightarrow}}

\newcommand{\letbang}{\mathsf{let}\,\oc}

\newcommand{\msec}{\mathsection}

\newcommand{\conf}[1]{\langle#1\rangle}

\newcommand{\unit}{\textbf{1}}
\newcommand{\Mtwo}{\textbf{2}}
\newcommand{\Mthree}{\textbf{3}}
\newcommand{\Mfour}{\textbf{4}}
\newcommand{\Mfive}{\textbf{5}}

\newcommand{\norm}[1]{\| #1\|}
\newcommand{\M}{\mathcal{M}}

\newcommand{\N}{\mathbb{N}}
\newcommand{\z}{\textbf{0}}
\newcommand{\class}[1]{\textbf{\textsc{#1}}}
\newcommand{\redt}[1]{\Downarrow^{#1}}
\newcommand{\Conf}{\mathsf{Conf}}
\newcommand{\MPi}{\Pi_\M}
\newcommand{\MLam}{\Lambda_\M}
\newcommand{\CR}{\mathsf{CR}}
\newcommand{\setp}[2]{\{\,\, #1 \,\,|\,\, #2 \,\,\}}
\newcommand{\dai}{\maltese}
\newcommand{\empStack}{\diamond}
\newcommand{\empEnv}{\diamond}
\newcommand{\empStore}{\diamond}

\newcommand{\LALrefmu}{\lambda_{\textsc{LAL}}^{\textsf{Reg},\mu}}

\newcommand{\interpR}[2]{|#1 |_{#2}}

\newcommand{\partitle}[1]{{\textbf{#1  \quad\textemdash\quad}}}
\newcommand{\Val}[1]{\llbracket #1 \rrbracket}

\newcommand{\uaff}{\lambda}
\newcommand{\upara}{\msec}
\newcommand{\ubang}{\oc}

\def\Mcon{\M[.]}
\newcommand{\StL}[2]{{#1}^{\leq #2}}

\newcommand{\preInter}[2]{\| R \vdash #1 \|_{#2}}
\newcommand{\Inter}[3]{|R\vdash #1 |_{#2}^{#3}}
\newcommand{\Stincl}{\sqsubseteq}

\newcommand{\STa}{{\mathcal{S}}}

\newcommand{\botR}{{\bot_\STa}}

\newcommand{\Vlist}[1]{\overline{#1}}

\newcommand{\letpar}{\mathsf{let}\,\msec}
\newcommand{\letdag}{\mathsf{let}\,\dagger}

\newcommand{\integer}[1]{\overline{#1}}

\newcommand{\unitv}{()}

 \title{Indexed realizability for bounded-time programming with
   references and type fixpoints} \titlerunning{Awesome paper}
\author{Alo\"is Brunel\inst{1} \and Antoine Madet\inst{2}}
\authorrunning{A. Brunel\and A. Madet} \institute{Alo\"is Brunel, Universit\'e Paris 13, Sorbonne Paris Cit\'e,\\ LIPN, CNRS, UMR 7030, F-93430, Villetaneuse, France.
 \and Univ Paris Diderot, Sorbonne Paris Cit\'e,\\ PPS, UMR 7126, CNRS,
F-75205 Paris, France}

\begin{document} 
\frontmatter  
\pagestyle{headings}
\maketitle

\begin{abstract}
  The field of implicit complexity has recently produced several
  bounded-complexity programming languages. This kind of language
  allows to implement exactly the functions belonging to a certain
  complexity class.  We here present a realizability semantics for a
  higher-order functional language based on a fragment of linear
  logic called \class{LAL} which characterizes
  the complexity class \class{PTIME}.  This
  language features recursive types and higher-order store.  Our
  realizability is based on biorthogonality, step-indexing and is
  moreover quantitative.  This last feature enables us not only to
  derive a semantical proof of termination, but also to give bounds on
  the number of computational steps needed by typed programs to
  terminate.
\end{abstract}

\section{Introduction}

	\partitle{Implicit computational complexity} This research field aims at 
		providing machine-independent characterizations of complexity classes 
		(such as polynomial time or logspace functions). 
		One approach is to use type systems based on linear logic to control the 
		complexity of higher-order functional programs. In
                particular, the so-called light
		logics (\emph{e.g.} \class{LLL}
                \cite{girard1995light}, \class{SLL} \cite{lafont2004soft}) have led to 
		various type systems for the
		$\lambda$-calculus guaranteeing that a well-typed term has a
		bounded complexity \cite{asperti1998light}. These
                logics introduce the modalities `$\oc$' (read
                \emph{bang}) and `$\msec$' (read \emph{paragraph}). By a fine
                control of the nesting of these modalities, which is
                called the \emph{depth}, the duplication of data can
                be made explicit and the complexity of programs can be
                tamed. This framework has been recently extended 
		to a higher-order process calculus \cite{DalLago2010Process} and a functional language with 
		recursive definitions \cite{baillot2010polytime}. Also, Amadio and Madet have 
		proposed \cite{AmadioMadet2011} a multi-threaded $\lambda$-calculus with higher-order store that
		enjoys an elementary time termination. \newline

	\partitle{Quantitative realizability} Starting from Kleene, the 
  concept of realizability has been introduced in different forms 
  and has been shown very useful to build models of computational 
  systems. In a series of works \cite{DalLago20112029,dal2010semantic}, 
  Dal Lago and Hofmann 
  have shown how to extend Kleene realizability with quantitative 
  informations in order to interpret subsystems of linear logic with 
  restricted complexity. The idea behind Dal Lago and Hofmann's work 
  is to consider bounded-time programs as realizers, where bounds are 
  represented by elements of a \emph{resource monoid}. 
  In \cite{brunel2012quantitative} the first author has shown how this quantitative
  extension fits well in a biorthogonality based framework (namely 
  Krivine's classical realizability \cite{krivine-realizability}) and how it
  relates to the notion of forcing.\newline
  
  \partitle{Step-indexing} In order to give a semantical account of
  features like recursive or reference types, one has to face
  troublesome circularity issues.  To solve this problem, Appel and
  McAllester \cite{AppelStepIndexed2001} have proposed step-indexed
  models. The idea is to define the interpretation of a type as a
  predicate on terms indexed by numbers. Informally, a term $t$
  belongs to the interpretation of a type $\tau$ with the index
  $k\in\N$ if when $t$ is executed for $k$ steps, it satisfies the
  predicate associated to $\tau$. Then, it is possible to define by
  induction on the index $k$ the interpretation of recursive or
  reference types.  Step-indexing has been  related to G\"odel-L\"ob
  logic and the \emph{later}  operator $\triangleright$ \cite{nakano2000modality}. \newline
 
  \partitle{Contributions} In this paper, we present a typed $\lambda$-calculus
  called $\LALrefmu$ whose functional core is based
  on the light logic \class{LAL} \cite{asperti1998light}.
  We extend it with recursive types and higher-order store. Even in presence of 
  these features, every program typable in $\LALrefmu$ terminates in polynomial
  time. To prove
  termination in bounded-time, we propose a new quantitative
  realizability semantics with the following features:
		\begin{itemize}
				\item It is biorthogonality based, which permits a simple presentation
				and allows the possibility to interpret control operators (though it is 
				only discussed informally in the conclusion of this paper).
				\item It is indexed, which permits
			to interpret higher-order store and recursive types. The particularity
			is that our model is indexed by depths
                        (the nesting of modalities) instead of computational steps (like in
			step-indexing).
		\end{itemize}
		To our knowledge, this is the first semantics presenting at the same
		time quantitative, indexed and biorthogonality features. \newline

 \partitle{Outline} Section \ref{sec:language} introduces the language
	 $\LALrefmu$ and its type system. In Section \ref{sec:indexed}, we introduce the indexed
 quantitative realizability. It is then used to obtain a semantic model for 
  $\LALrefmu$, which in turn implies termination in polynomial time of
  typed programs. 
 Finally, we mention related works in Section \ref{sec:related} and
 in Section \ref{sec:conclusion} we discuss future research directions and conclude.

\section{The language}
\label{sec:language}
 This section presents the language
$\LALrefmu$ and its type system. Before going into details, we
give some intuitions on the modalities bang and paragraph and explain
how we deal with side-effects with the notion of \emph{region}.\newline

\partitle{On bang and paragraph}
The functional core of the language is an affine $\lambda$-calculus
which means that functions use their argument at most
once.  Modal constructors `$\oc$' and `$\msec$' originating from
\class{LAL} are added to the language. Intuitively, from a value $V$ which can be used at most
once, we can construct a value $\oc V$ which can be duplicated with a
special modal binder. Values which have been duplicated are of the shape
$\msec V$ and therefore cannot be duplicated anymore. We will see more
precisely with the type system how we can be polynomial with these modal operators. \newline

\partitle{On regions}
Following a standard practice in effect systems~\cite{Lucassen88}, the
global store is abstracted into a finite set of regions where each region represents one or several
dynamic memory locations. Then, side-effects are produced by read and write
operators on constant region names. 
As noted by Amadio \cite{amadio2009stratified}, the abstract language
with regions simulates the concrete language with dynamic memory
locations as long as the values assigned to regions do not erase
the previous ones. In particular, termination in polynomial time of the language with
regions should entail termination in polynomial time of the language with
references. There are two reasons for working with regions instead of memory
locations. First, they allow to deal with several kinds of
side-effects. Our language is sequential hence regions naturally represent
higher-order references \emph{\`a la} \textsf{ML}, but in the context
of concurrent programming they could represent communication channels
 or even signals in the context of  reactive
programming. Second, we find it easier to give a semantic model of a type
system with regions instead of dynamic addresses.

\subsection{Syntax and operational semantics}
The syntax of the language is the following:
$$
\begin{array}{l@{\quad\quad}rcll}
Values 		& V 	&	::=	& x \mid \lambda x.M  \mid r
\mid \unitv \mid \integer{n} \mid \oc V \mid \msec V\\
Terms 		& M 	&	::=	& V \mid M_1 M_2 \mid  V_1
\star V_2 \mid \oc
M \mid \msec M\\
&&&\letb{x}{V}{M} \mid \letp{x}{V}{M}\\
					&			&     &	\get{r} \mid \st{r}{V} 
\end{array}
$$  
We suppose having a countable set of variables denoted $x,y,\ldots$ and of regions
denoted by the letters $r,r',\ldots$. The terminal value \emph{unit}
is denoted by $\unitv$. Integers are denoted by $\integer{n}$ and $V_1
\star V_2$ stands for any arithmetical operation. Modal terms and modal values  are
built with the unary constructors $\oc$ and $\msec$ and are destructed
by the respective $\letbang$ and $\letpar$ binders. The terms
$\get{r}$ and $\st{r}{V}$ are respectively used to read a value from a
region $r$ and to assign a value $V$ to a region $r$. As usual we
write the sequential composition $M;N$ for $(\lambda x.N)M$ where $x$
does not occur free in $N$. We denote by $\sub{M}{N/x}$ the term $M$ in which each free occurrence of 
$x$ has been substituted by $N$.\newline

The operational semantics of the language is presented in the form of
an abstract machine. We first define the configurations of the abstract
machine:
$$
\begin{array}{l@{\quad\quad}rcl}
Environments 	& E 	&	::=	& \empEnv \mid V \cdot E \mid
M \odot E \mid \oc \cdot E \mid \msec \cdot E\\
Stores 		& S 	&	::=	& \store{r}{V} \mid S_1 \uplus S_2\\
Configurations & C & ::= & \conf{M,E,S}
\end{array}
$$
Programs are intended to be executed with a right-to-left
call-by-value strategy. Hence, an environment
$E$ is either an empty frame $\empEnv$, a stack of frames to
evaluate on the left of a value ($V \cdot E$),
on the right of a term ($M \odot E$) or in-depth of a term ($\oc
\cdot E$ and $\msec
\cdot E$).
Finally, a store $S$ is a multiset of
region assignments $\store{r}{V}$.
A configuration of the abstract machine is executed according to the
following rules:
$$
  \begin{array}{rcl}
    \conf{\integer{n_1} \star \integer{n_2},E,S} & \reduc	&
    \conf{\integer{n_1 \star n_2},E,S}\\
    \conf{\lambda x.M,V \cdot E,S} & \reduc	&	\conf{\sub{M}{V/x},E,S}\\
    \conf{MN,E,S} 									& \reduc	&	\conf{N,M \odot E,S}\\
    \conf{V,M \odot E,S} 						& \reduc	&	\conf{M,V \cdot E,S}\\
    \conf{\dagger M,E,S}  								& \reduc	& \conf{M,\dagger \cdot E,S} \qquad
    \textrm{if } M \textrm{ is not a value}\\
    \conf{V,\dagger \cdot E,S} 					& \reduc	& \conf{\dagger V,E,S}\\
    \conf{\letd{x}{\dagger V}{M},E,S} 			&	\reduc	& \conf{\sub{M}{V/x},E,S}\\
    \conf{\get{r},E,\store{r}{V} \uplus S}	&	\reduc	&	\conf{V,E,S}\\
    \conf{\st{r}{V},E,S} 										&	\reduc	& \conf{\unitv,E,\store{r}{V} \uplus S}\\
  \end{array}
$$  
For the sake of conciseness we wrote
$\dagger$ for $\dagger \in \set{\oc,\msec}$.
Observe that the $\letdag$-binders destruct modal values $\dagger V$
and propagate $V$. Therefore, a value $\oc^n V$ should  be
duplicable at least $n$ times. Reading a region amounts to
\emph{consume} the value from the store and writing to a region
amounts to \emph{add} the value to the store. We consider programs up to $\alpha$-renaming and in the
sequel $\treduc$ denotes the reflexive and transitive closure of
$\reduc$.

\begin{example}
\label{example-f}
Here is a function
$
F = \lambda x.\letb{y}{x}{\st{r_1}{\msec y};\st{r_2}{\msec y}}
$
that duplicates its argument and assign it to regions $r_1$
and $r_2$. It can be used to duplicate a value from another region
$r_3$ as follows:
$$
\conf{F\get{r_3},\empEnv,\store{r_3}{\oc
    V}} \treduc \conf{\unitv,\empEnv,\store{r_1}{\msec
    V}\uplus\store{r_2}{\msec V}}
$$   
Therefore the value $\msec V$ stored in $r_1$ and $r_2$ is no longer duplicable.
\end{example}

\begin{definition}
  We define the notation $\conf{M,E,S} \redt{n}$ as the
  following statement:
  \begin{itemize}
  	\item The evaluation of $\conf{M,E,S}$ in the abstract
  	machine terminates.
  	\item The number of steps needed by $\conf{M,E,S}$ to
  	terminate is $n$.
  \end{itemize}
\end{definition}

\subsection{Type system} 
The light logic \class{LAL} relies on a \emph{stratification} principle
which is at the basis of our type system. We first give an informal explanation
of this principle.\newline

\partitle{On stratification}
Each occurrence of a program can be given a depth which is
the number of nested modal constructors that it appears under. Here is
an example for the program $P$ where each occurrence is labeled
with its depth:
$$
\begin{array}{c}
  \begin{minipage}{0.5\linewidth}
$
  \begin{array}{r}
{\small \text{$ P = (\lambda x.\letb{y}{x}{\st{r}{\msec y}})\oc V;\get{r}$}   }
  \end{array}    
$
  \end{minipage}\\\\
  \begin{minipage}{0.8\linewidth}
    \begin{center}
      
      \begin{tikzpicture}[grow=right,level distance=1.2cm,font=\scriptsize,sibling distance=0.7cm]
      \node  {$;^{0}$}
        child {node {$@^0$}
          child {node {$\lambda x^{0}$}
            child {node {$\letbang y^{0}$}
              child {node {$x^{0}$}}
              child {node {$\mathsf{set}(r)^{0}$}
                child {node {$\msec^{0}$}
                  child {node {$y^{1}$}}}}}}
          child {node {$\oc^{0}$}
            child {node {$V^{1}$}}}}
        child {node {$\get{r}^0$}};
    \end{tikzpicture}  
    \end{center}

  \end{minipage}
\end{array}
$$
The depth $d(M)$ of a term $M$ is the maximum depth of its occurrences.
The stratification principle is that the depth of every occurrence is
preserved by reduction. On the functional side, it can
be ensured by these two constraints: (1)
if a $\lambda$-abstraction occurs at depth $d$, then it binds at depth
$d$; (2) if a $\letdag$ occurs at depth $d$, then it binds at depth
$d+1$. These two constraints are respected by the program
 $P$ and we observe in the following reduction 
$$
\conf{P,\empEnv,\emptyset} \treduc
\conf{\st{r}{\msec V}, \empEnv,\emptyset}
$$
that the depth of $V$ is preserved. In order to preserve the depth of
occurrences that go through the store, this third constraint is needed:
(3) for each region $r$, $\get{r}$ and $\mathsf{set}(r)$ must occur at
a fixed depth $d_r$. We observe that this is the case of program $P$
where $d_r = 0$. Consequently, the reduction terminates as follows
$$
\conf{\st{r}{\msec V},\oc \cdot \empEnv,\emptyset} \treduc
\conf{\msec V,\empEnv,\emptyset}
$$
where the depth of $V$ is still preserved. Stratification on the
functional side has been deeply studied by Terui with the Light Affine
$\lambda$-calculus \cite{Terui07} and extended to regions by
Amadio and the second author \cite{AmadioMadet2011}.\newline

We now present the type system for $\LALrefmu$ that formalizes
stratification and should ensure polynomial soundness.\newline

\partitle{Types and contexts} The syntax of types and contexts is the following:
$$
\begin{array}{l@{\quad\quad}rcl}
Types & A,B&::=& \alpha \mid \tertype \mid \nat \mid  A\multimap B \mid \oc A \mid \msec A

\mid \mu X.A \mid \rgtype{r}{A}   \\
Variable\; contexts & \Gamma,\Delta &::=&
\hyp{x_{1}}{u_{1}}{A_{1}},\ldots,\hyp{x_{n}}{u_{n}}{A_{n}} \\
Region \;contexts & R&::=& r_1:(\delta_1,A_1),\ldots,r_n:(\delta_n,A_n)  
\end{array}
$$
We have a countable set of type variables $X,X',\ldots$
Then, we distinguish the terminal type
$\tertype$, the type of integers $\nat$, the affine functional type $A \limp B$, the type
$\bang A$ of values that can be duplicated, the type
$\msec A$ of values that may have been duplicated, recursive types $\mu X.A$ and the type $\rgtype{r}{A}$ of 
regions $r$ containing values of type $A$. Hereby types may depend on regions.
Following~\cite{AmadioMadet2011}, a region context associates a
natural number $\delta_i$ to each region $r_i$ of a finite set of
regions $\set{r_1,\ldots,r_n}$ that we write $dom(R)$. Writing
$r:(\delta,A)$ means that the region $r$ contains
values of type $A$ and that gets and sets on $r$ may only happen at a
fixed depth depending on $\delta$.
A variable context associates each variable
with an usage $u \in \set{\uaff,\upara,\ubang}$ which constraints the
variable to be bound by a $\lambda$-abstraction, a
$\textsf{let}\,\msec$-binder or a $\textsf{let}\,\oc$-binder
respectively. In the sequel we write $\Gamma_u$ for
$x_{1}:(u,A_1),\ldots,x_{n}:(u,A_n)$. Writing $x : (u, A)$ means that
the variable $x$ ranges on values of type $A$ and can be bound
according to $u$.\newline

Types depend on region names. As we shall see, this allows for a
straightforward interpretation of the type $\regtype{r}{A}$. Moreover,
it induces a typed translation from a language with dynamic locations
to a language with regions. For example, for every occurrence $P$ of
dynamic allocation like $\lets{x}{\mathsf{ref}\, M}{N}$
in an \textsf{ML} program where $M$ is of type $A$,  
it suffices to introduce a distinct region name $r$ and associate the
variable $x$ with type $\regtype{r}{A}$. Then, the type preserving translation
of $P$ is simply $\st{r}{M};\sub{N}{r/x}$.
However, we have to be careful in stating when a type $A$ is well-formed with respect to a region context $R$, written $R \vdash A$. Informally, the judgment
$r_{1}:(\delta_1,A_{1}),\ldots,r_{n}:(\delta_n,A_{n}) \vdash B$ is well formed provided
that:
$(1)$ all the region names
occurring in the types $A_1,\ldots,A_n,B$ belong to the set
$\set{r_1,\ldots,r_n}$, 
$(2)$ all types of the shape
$\rgtype{r_{i}}{B}$ with $i\in \set{1,\ldots,n}$ and occurring in the
types $A_1,\ldots,A_n,B$ are such that \mbox{$B=A_i$}. The judgment
$R \vdash \Gamma$ is well-formed if $R \vdash A$ is well-formed for
every $x:(u,A) \in \Gamma$. We invite the reader to check in
\cite{amadio2009stratified} that these judgements can be easily defined.\newline

\partitle{Typing rules}
A typing judgment takes the form
$
R;\Gamma \vdash^{\delta} P:A
$ and is indexed by an integer $\delta$. The rules  are given in Figure~\ref{common-system}. They  should entail the following:
 \begin{itemize}
 \item if $x:(\lambda,A) \in \Gamma$ then
   $x$ occurs at most once at depth $0$ in $ P$,
 \item if $x:(\msec,A) \in \Gamma$ then
   $x$ occurs at most once at depth $1$ in a subterm $\msec M$ of $P$,
 \item if $x:(\oc,A) \in \Gamma$ then
   $x$ occurs arbitrarily many times at depth $1$ in a subterm $\msec
   M$ or $\oc M$ of
   $ P$,
\item  if $r:(\delta',A) \in R$ then $\get{r}$ and $\mathsf{set}(r)$ occur at depth
$\delta-\delta'$ in $P$.
 \end{itemize}
\begin{figure}[ht]
$$
\begin{array}{c}
  \inference
[{\small \textsf{int}}]
  {R\vdash}
  {R;-  \vdash^\delta \integer{n}:\nat}

\qquad 

  \inference
[\textsf{arith}]
  {R;\Gamma \vdash^\delta V_1 : \nat & R;\Delta \vdash^\delta V_2 : \nat}
  {R;\Gamma,\Delta  \vdash^\delta V_1 \star V_2:\nat}

\\\\
  \inference
[\textsf{v}]
  {R\vdash}
  {R;x:(\uaff,A)  \vdash^\delta x:A}
\quad
\inference
[\textsf{u}]
{R\vdash }
{R;- \vdash^\delta \unitv:\tertype}

\quad

\inference
[\textsf{r}]
{R\vdash   &  r:(\delta,A) \in R}
{R;-  \vdash^\delta r:\rgtype{r}{A}}\\ \\ 

\inference
[\textsf{c}]
{R;\Gamma,x:(\ubang,A),y:(\ubang,A) \vdash^\delta M : B}
{R;\Gamma,z:(\ubang,A) \vdash^\delta \sub{M}{z/x,z/y} : B}

\;

\inference
[\textsf{w}]
{R;\Gamma \vdash^\delta M : B & R\vdash\Gamma,x:(u,A)}
{R;\Gamma,x:(u,A) \vdash^\delta M : B}
\\ \\

  \inference
  [\textsf{lam}]
  {R;\Gamma,x:(\uaff,A)  \vdash^\delta M:B}
  {R:\Gamma \vdash^\delta \lambda x.M: A \multimap B}

  \;

  \inference
  [\textsf{app}]
  {R;\Gamma \vdash^\delta M_1:A\multimap B & R;\Delta \vdash^\delta
    M_2:A}
  {R;\Gamma,\Delta \vdash^\delta M_1M_2:B}
\\ \\

    \inference
    [\textsf{$\oc$-prom}]
  {R;x:(\uaff,A) \vdash^{\delta} V:A }
  {R;x:(\ubang,A) \vdash^{\delta+1} \oc V:\oc A}

  \;

  \inference
[\textsf{$\msec$-prom}]
  { R;\Gamma_{\uaff},\Delta_{\uaff} \vdash^{\delta} M:A}
  { R;\Gamma_{\upara},\Delta_{\ubang} \vdash^{\delta+1} \msec M: \msec
    A}
\\\\
  \inference
  [\textsf{$\dagger$-elim}]
  {R;\Gamma \vdash^\delta V:\dagger A \\ R;\Delta,x:(\dagger,A) \vdash^\delta
    M:B}
  {R;\Gamma,\Delta \vdash^\delta \letd{x}{V}{M}:B}

  \;





\inference
[\textsf{get}]
{R;-\vdash^\delta r : \regtype{r}{A} }
{R;-  \vdash^\delta \get{r}:A}
\\\\\

\inference
[\textsf{set}]
{R;-\vdash^\delta r : \regtype{r}{A}\\ R;\Gamma  \vdash^\delta
V : A}
{R;\Gamma  \vdash^\delta \st{r}{V}:\tertype} 



\;

  \inference
  [\textsf{un/fold}]
  {R;\Gamma \vdash^\delta M:\mu  X .A}
  {R;\Gamma \vdash^\delta M:\sub{A}{\mu X .A/ X}}




\end{array}
$$
\caption{Typing rules}
\label{common-system}
\end{figure}

Several remarks have to be made:
\begin{itemize}
\item In binary rules, we implicitly require that contexts $\Gamma$
  and $\Delta$ are disjoints. They are explicit rules for the
  weakening and contraction of variables and we may only
  contract variables with usage $\oc$. Therefore, $\letbang$s are the
  only operators that can bind several occurrences of a variable.
\item There are  two restrictions to apply the rule $\oc$-\textsf{prom}:
   first, $V$ may contain at most one
  occurrence of a free variable; second, $V$ is a value (we cannot
  type $\oc M$). The first constraint ensures that the size of a
  program does not explode exponentially and is well studied in
  \cite{Terui07}. The second condition is due to side-effects. To see
  this, assume that we can type $\oc M$. Then we can derive the judgment 
  $
  r:(0,A);-\vdash^1 \lambda x.\letb{y}{x}{\msec \st{r}{x}; \oc\get{r}}
    : \msec A \multimap \oc A
    $.
    Both $\get{r}$ and $\mathsf{set}(r)$ occur at depth $1$ but under
    different modalities. Clearly the type $\msec A \multimap \oc A$
    which cannot be derived in \class{LAL} has to be rejected  for otherwise we
    can freely re-duplicate a value $\msec V$.
\item The depth $\delta$ of a judgment is incremented when we
  construct a modal term. This allows to count the number of nested
  modalities and to stratify regions  by
  requiring  that the depth of a region matches the depth of the
  judgment in the rule \textsc{r}.
\item For space consideration the rule \textsf{un/fold} can be used
  upside down. 
\end{itemize}

\begin{definition}
  We say that a program $M$ is well-typed if a judgment
  $R;\Gamma\vdash^\delta M:A$ can be derived for some $R$, $\Gamma$
  and $\delta$ such that:
  \begin{itemize}
  \item If $r:(\delta_r,A) \in R$ then $A = \msec B$.
  \item For every type fixpoint $\mu X.A$ that appears in $R$ and
    $\Gamma$, the occurrences of $X$  in $A$ are \emph{guarded} by (occur
    under) a
    modality $\dagger$.
  \item Every depth index in the derivation is positive. Note that if
    this is not the case, we
    can always find  $\delta ' > \delta$ such that this is true for  $R;\Gamma\vdash^{\delta'} M:A$.
  \end{itemize}
These three conditions will be needed to give a well-founded interpretation.
\end{definition}

\begin{example}
The operational semantics of references (values are \emph{copied} from
the store) can be simulated as long as values stored in regions are of
the shape $\oc V$. For example, consider the following well-typed program
$$
r:(0,\nat);-\vdash^1 \letb{x}{\get{r}}{\st{r}{\oc x}};
\msec(x * x) : \nat
$$
that reduces as follows:
$$\conf{\letb{x}{\get{r}}{\st{r}{\oc x }};
\msec(x * x),\empEnv,\store{r}{\oc\integer{n}}} \treduc
\conf{\msec(\integer{n * n}),\empEnv,\store{r}{\oc\integer{n}}}
$$
Indeed, the region $r$ can be considered a \emph{reference} since
the value $\oc\integer{n}$ has not been erased from the store. 
\end{example}

The following progress property can be derived as long the program do
not try to read an empty region.
\begin{proposition}[Progress]
  If $R;\Gamma \vdash^\delta M : A$ then $\conf{C,\empEnv,\emptyset}
  \treduc \conf{V,\empEnv,S}$ and $R;\Gamma \vdash^\delta V : A$ and
  every assigned value in $S$ can be typed.
\end{proposition}

The goal of the next section is  to prove the following theorem
 \begin{theorem}[Polynomial termination]
\label{thm:termination}
 	 There exists a family of polynomials $\{P_d \}_{d\in\N}$
 	 such that if $M$ is well-typed then
   $\conf{M,\empStack,\empStore}$ terminates in at most
   $P_{d(M)}(size(M))$ steps.
 \end{theorem}

\section{``Indexed'' quantitative realizability}
\label{sec:indexed}
	We now present a biorthogonality-based interpretation of
$\LALrefmu$.  Apart from the use of biorthogonality, this interpretation
has two particularities:

\begin{itemize}
	\item	First, the realizability model is \emph{quantitative}.
	A type is interpreted by a set of
	\emph{weighted realizers} (that is a program together with 
	a store and a quantity bounding its normalization time). This allows
	to prove complexity properties of programs.
	\item Secondly, the semantics is \textit{indexed}
	(or \textit{stratified}), meaning that we interpret
	a type by a family of sets indexed by $\N$. Moreover
	the interpretation of a type is defined by double induction,
	first on the \emph{index} $n$, and secondly on the \emph{size} of the type. 
	This allows to interpret recursive types and references.
\end{itemize}

It is worth noticing that while our interpretation is similar
to the so-called "step-indexed" models, the meaning of indexes
is not (directly) related to the number of computation steps but to
the depth of terms (and so our model could be described
as a "depth-indexed" model). It is the quantitative part 
which is used to keep track of the number of computational steps.

\subsection{The light monoid}

The realizability model is parametrized by a monoid, whose
elements represent an information about the 
amount of time needed by a program to terminate. 
To interpret $\LALrefmu$, we use the \emph{light monoid}, 
which is a simplification of a \emph{resource monoid}
introduced in \cite{dal2010semantic}. 

\begin{definition} The \emph{light monoid} is defined as
 the structure $(\M, +, \norm{.}, !,\S)$ where
 \begin{itemize}
  \item $\M$ is the set of triples $(n,m,f)$ such that $n,m\in \N$
		and $f : \N\rightarrow \N$ is increasing.
  \item If $(n,m,f), (l,k,g)\in\M$, 
		$(n,m,f) + (l,k,g)  = (n+l, max(m,k), max(f,g))$.
  \item If $(n,m,f)\in\M$, $\norm{(n,m,f)}= n\,f(m+n)$.
  \item If $(n,m,f)\in\M$,  $\S(n,m,f) =  (n/m, m, x\mapsto x^2 f(x^2))$.
  \item Finally, for any $(n,m,f)\in\M$, $!(n,m,f) = (1, n+m, x\mapsto x^3 f(x^3))$.
	\end{itemize}
	We moreover denote by $\textbf{n}$ the element of $\M$ defined
	as $(n, 0, x\mapsto 0)$. 
\end{definition} 

From now on, we use lower-case consonnes letters $p,q,m,v,\dots$ to denote 
elements of $\M$.
To give some intuitions on these operations, let's say that
$+$ represents the resource consumption resulting of the
interaction of two programs, and that $\norm{p}$ calculates
the \emph{concrete} bound (a natural number) associated to an abstract bound
$p\in\M$.

\begin{remark}
The structure $(\M, +, \z, \norm{.})$ is a quantitative
monoid in the sense of \cite{brunel2012quantitative}. In particular,
it satisfies the following inequality:
$$\forall p, q,\norm{p}+\norm{q}\leq \norm{p+q}$$
This inequality informally represents the fact that the amount
of resource consumed by the interaction of two programs is 
more than the total amount of resource used by the two
programs alone.
\end{remark}

\begin{definition}
Given $p,q\in \M$, we say that $p< q$ iff $\forall r\in \M, \norm{p+r}< \norm{q+r}$
and $p \leq q$ iff $\forall r\in \M, \norm{p+r} \leq \norm{q+r}$
\end{definition}

\begin{property}
The relation $\leq$ enjoy the following properties:
\begin{itemize}
  \item If $p\leq q$, then $p+r\leq q+r$.
  \item If $p\leq q$, then $\norm{p}\leq \norm{q}$.
  \item If $p\leq p'$ and $q\leq q'$, then $p+q\leq q+q'$.
\end{itemize}
\end{property}

\begin{property}\label{prop:monoid}
The operations $!$ and $\S$ on the monoid $\M$ satisfy the 
following properties:
\begin{itemize}
\item $\S p \leq ! p$ 
\item $\S (p+q) \leq \S p+\S q$
\item $\S p +\S q  \leq \S (p+q) + \Mtwo$
\item $!p+!p = !p + \unit$
\end{itemize}
\end{property}

A third operation $F$ will be used to interpret functoriality of $!$:
\begin{property}\label{prop:functoriality}
 Let $p = (n,m,f)$ and $q$ two elements of $\M$. 
 We define $F(p)$ to be $(1+n+m, m, x\mapsto x^3 f(x^3))$.
 We then have
 $!(p+q) \leq F(p) + !q$.
\end{property}

Finally we define a notion of $\M$-context, which is similar to an evaluation context, 
but for resource bounds instead of programs. 

\begin{definition}
A $\M$\emph{-context} is any function $f : \M \rightarrow \M$
obtained by composition of the functions
$\lambda x.x+p$ (where $p$ is any element of $\M$),
$\lambda x.! x$ and $\lambda x.\S x$.
The set of $\M$-contexts is denoted by $\Mcon$.
\end{definition}

\begin{example}
For instance, $\lambda x.(!(!\S x+\unit))$ is a valid $\M$-context,
but $\lambda x.(!x+x+\S x)$ and $\lambda x.\z$ are not.
\end{example}

Any $\M$-context $f$ is monotonic in its argument, that is if $p < q$, then $f(p)< f(q)$ and
if $p \leq q$, then $f(p) \leq f(q)$. 

\subsection{Orthogonality}

Usually, orthogonality is defined between a program
and an environment. Here, it defined between
a weighted program and a weighted environment.

\begin{definition}
\begin{itemize}
  \item A \textit{weighted term} is a tuple $(M,p)$ where
  $M$ is a \emph{closed} term and $p$ an element of $\M$. The set of weighted terms
  is denoted by $\MLam$.
  \item A \textit{weighted stack} is a pair $(E, e)$ where
  $E$ is a \emph{closed} stack and $e$ an element of $\Mcon$. The set of weighted stacks
  is denoted by $\MPi$.
\end{itemize}
\end{definition}

We choose a \emph{{pole}} $\Bot \subseteq \Conf \times \M$ as the set
of bounded-time terminating weighted configurations:
 $$\Bot = \set{(\conf{M,E,S},p) \mid \conf{M,E,S} \redt{n}\wedge n \leq \norm{p}} $$
In orthogonality-based models, fixing a pole, also called \emph{observable}, corresponds 
to choose a notion of \emph{correct} computation. 

 \begin{proposition}This pole satisfies some important properties:
   \begin{enumerate}
   \item ($\leq$-saturation) If $(\conf{M,E,S},p) \in \Bot$ and $p \leq q$ then $(\conf{M,E,S},q) \in \Bot$.
   \item ($\reduc$-saturation) If $(\conf{M,E,S},p) \in \Bot$ and
     $\conf{M',E,S'} \reduc \conf{M,E,S}$ then $(\conf{M',E',S'},p+\unit) \in \Bot$.
   \end{enumerate}
 \end{proposition}
 
The pole induces a notion of orthogonality.
In contrast with usual models, since we need to deal with references,
the orthogonality relation is parametrized by a set $\STa$ of stores. 

\begin{definition}
The \emph{orthogonality relation} $\botR \subseteq \MLam\times\MPi$ is defined as:
$$(M,p) \botR (E,e)\text{ iff } \forall (S,s)\in \STa, (\conf{M,E,S},e(p+s)) \in \Bot$$
This orthogonality relation lifts to sets of weighted terms and weighted stacks.
If $X \subseteq \MLam$ (resp $X\subseteq \MPi$),
$$X^{\botR} = \setp{(E,e)\in \MPi}{\forall (M,p) \in X, (M,p) \botR (E,e)}$$
$$\text{(resp. } X^{\botR} = \setp{(M,p) \in \MLam}{\forall (E,e) \in X, (M,p) \botR (E,e)}
\text{ )}$$
\end{definition}

The operation $(.)^\botR$ satisfies the usual orthogonality properties.
   \begin{lemma}
     \label{lem-ortho}
     Suppose $X,Y \subseteq \MLam$ or $X,Y \subseteq\MPi$:
     \begin{enumerate}
     \item $X \subseteq Y$ implies $Y^{\botR} \subseteq X^{\botR}$
     \item $X \subseteq X^{\botR\botR}$
     \item $X^{\botR\botR\botR} = X^{\botR}$
     \end{enumerate}
   \end{lemma}

\begin{definition}
If $X$ is a set of weighted realizers, we define its \emph{$\leq$-closure} 
$\overline{X} = \setp{ (M,p)}{ \exists q \leq p, (M,q)\in X}$.
\end{definition}
\begin{remark}
Notice that for any $\STa$, we have $\overline{X} \subseteq X^{\bot_\STa\bot_\STa}$. 
\end{remark}

We say that a set $X\subseteq \MLam$ is a 
$\STa$-behavior if $X = X^{\botR\botR}$.
Finally, we can define the set of $\STa$-reducibility candidates. To do that, we
first need to extend the language of terms with a new constant
$$M \quad ::= \quad \mid \dai$$
This constant comes with no particular reduction rule. It can be seen as a 
special variable considered as a closed term and is in a sense the dual 
of the empty stack. 
 \begin{definition}
   The set of $\STa$\emph{-reducibility candidates}, denoted by $\CR_\STa$ is 
   the set of $\STa$-behaviors $X$ such
   that $(\dai, \z) \in X \subseteq \{(\empStack, x\mapsto x)\}^\botR$
 \end{definition}
\begin{remark}
If $(M,p)\in X$ where $X$ is a $\STa$-reducibility candidate
and if $(\empStore, \z)\in \STa$, then
$\conf{M, \empStack, \empStore}$ terminates in at most $\norm{p}$
steps. In fact our notion of reducibility candidate extends the
usual notion in the non-quantitative case.
\end{remark}
 
Finally, suppose $R$ is a set of regions and suppose $\STa_R$ is a set of stores 
whose domain is restricted to a $R$. We say that :

\quad\quad $\STa_R\Stincl \STa' \Leftrightarrow $ $\STa'$ contains $\STa_R$ and if $(S,s)\in \STa'$ and if we write $S = S^\delta \uplus S''$, then there is a decomposition $s=s'+s''$ such that
$(S^\delta,s') \in \STa_R$ and moreover, if $(S_R,s_R)\in \STa_R$ then $(S''\uplus S_R, s''+s_R)\in \STa'$.

\begin{remark}
This quite involved definition will permit to the interpretation 
of a type to enjoy properties similar to the one
called \emph{extension/restriction} in \cite{amadio2009stratified}.
In other words, given a store, it gives a way to say what substore
can be removed safely and what stores can be added to it safely.
\end{remark}

\subsection{Interpretation of $\LALrefmu$}

Using the orthogonality machinery previously defined, we
can give an interpretation of $\LALrefmu$ types
as reducibility candidates. Suppose $R$ is the following region context:
   $$R=\hyp{r_1}{\delta_1}{\S  A_1},\ldots,\hyp{r_n}{\delta_n}{\S A_n}$$  
 We define three indexed sets: the \emph{interpretation} $\interpR{R}{\delta}$ of the region 
 context $R$, the \emph{pre-interpretation} $\preInter{A}{\delta}$ of a type $A$ and its
 \emph{interpretation} $\Inter{A}{\delta}{\STa}$ with respect to a set of stores
 $\STa$.  These three notions are defined by mutual induction, first on 
 the index $\delta$, and then on the size of the type $A$. 
   $$
	   \begin{array}{rcl}
	   \interpR{R}{=\delta} & = & \setp{(S,{\displaystyle
	   			\sum_{\delta_i = \delta}
	        \sum_{\substack{1 \leq j \leq k_i}}} \S q^i_j)}{dom(S) = \setp{r_i}{r_i : (\delta_i, \S A_i) \in R\wedge \delta_i= \delta} \\
	        & & \wedge\quad \forall r_i \in dom(S),  
	         					S(r_i) = \set{\S V_1^i,\S V_2^i,\ldots,\S V_{k_i}^i}\\
	       	& & \wedge\quad \forall j \in [1,k_i], 
	       		(V_j^i,q_j^i) \in \overline{\preInter{A_i}{\delta_i-1}}}\\
	       	& & \\
	   \interpR{R}{\delta+1} & = & \setp{(S,s)}{\exists (S_1,s_1)\in \interpR{R}{=\delta+1},
	   \exists (S_\delta,s_\delta) \in \interpR{R}{\delta}, 
	    S = S_1 \uplus S_\delta \wedge s = s_1 +\S s_\delta}
	 	 \end{array}
 	 $$
 	 For convenience, 
we start the indexing of the interpretation at $-1$ instead of $0$.
	\begin{eqnarray*}
		\preInter{A}{-1} & = & \{ (\dai, \z)\}
	\end{eqnarray*}
	For $\delta \geq 0$, we define the pre-interpretation as:
	\begin{eqnarray*}
  	\preInter{\nat}{\delta} 			 & = & \setp{(\integer{n},\z)}{n\in\N}\\
  	\preInter{\tertype}{\delta} 			 & = & \{(\unitv,\z)\}\\
	  \preInter{\regtype{r}{A}}{\delta}  & = & \{(r,\z)\}\\
   	\preInter{A \multimap B}{\delta}   & = & \setp{(\lambda x.M,p)}{
   			 \forall (V,v) \in \preInter{A}{\delta},
   		   \forall \STa, \interpR{R}{\delta} \Stincl \STa, 
       	 (\sub{M}{V/x},p+v) \in \Inter{B}{\delta}{\STa}}\\
    \preInter{\S A}{\delta} 				   & = & \setp{(\S V,\S v)}{(V,v)\in\preInter{A}{\delta-1}}\\
    \preInter{\oc A}{\delta}    			 & = & \setp{(\oc V, \oc v)}{(V,v)\in\preInter{A}{\delta-1}}\\
    \preInter{\mu X.A}{\delta} 					& = & \preInter{A[\mu X.A/X]}{\delta}
  \end{eqnarray*}
  The interpretation of a type with respect to a set $\STa$ is just defined
  as the bi-orthogonal of the pre-interpretation:
    $$\Inter{A}{\delta}{\STa} = \preInter{A}{\delta}^{\bot_{\STa}\bot_{\STa}}$$
      
\begin{remark}
Because of the presence of type fixpoints and regions,
there are several circularities that could appear
in the definition of $\preInter{A}{\delta}$. Yet, 
the interpretation is well defined for the following reasons:
	\begin{itemize}	
		\item The type fixpoints $\mu X.A$ we consider are such that 
		every occurrence of $X$ in $A$ is \emph{guarded} by a modality
		$!$ or $\S$. But these modalities make the index of the interpretation
		decrease by one.
		Hence, $\preInter{\mu X.A}{\delta+1}$ is
		well defined as soon as $\preInter{\mu X.A}{\delta}$ is. 
		\item To define $\preInter{A}{\delta+1}$, we need $\interpR{R}{\delta+1}$ 
		to be already defined. But here again, in $R$ each type is guarded
		by a modality $\S$. This implies that to define $\interpR{R}{\delta+1}$, we
		only need to know each $\preInter{A_i}{\delta}$.
	\end{itemize}
\end{remark}

An important point is that 
the interpretation of a formula $A$ with respect to a region context $R$ and
to an index $\delta \in \N$ is a $\interpR{R}{\delta}$-reducibility candidate (it
will be used to prove bounded-time termination). 

 \begin{proposition}
\label{prop-redcan}
   For all $\delta \in \N$ we have $\Inter{A}{\delta}{\interpR{R}{\delta}} \in \CR_{\interpR{R}{\delta}}$.
 \end{proposition}

This is to prove the inductive case of the arrow $\multimap$ that the use of the constant $\dai$ is mandatory.
Indeed it is used in the proof to reduce \emph{under the} $\lambda$s.
 
 \subsection{Adequacy and bounded-time termination}

 \begin{table}[h]
   $$
   \begin{array}{c}
     \inference[\textsf{v}]
     {}
     {\vdash^{\delta} x : \z}
     \qquad
     \inference[\textsf{r}]
     {}
     {\vdash^{\delta} r : \z}
     \qquad
      \inference[\textsf{u}]
      {}
      {\vdash^{\delta} \unitv :\z}
      \\\\
      \inference[\textsf{int}]
      	{}
      	{\vdash^\delta \integer{n} : \z}
      \qquad
      \inference[\textsf{arith}]
      {\vdash^\delta V_1 : \Val{V_1}\quad\quad \vdash^\delta V_2 : \Val{V_2}}
      {\vdash^\delta V_1 \star V_2 : \Val{V_1} + \Val{V_2} }
      \qquad
     \inference[\textsf{get}]
			{}
			{\vdash^\delta \get{r}: \Mfive}\qquad
      \\\\
		 \inference[\textsf{set}]
			{\vdash^{\delta}  V : \Val{V}}
			{\vdash^{\delta+1} \st{r}{\S V}: \S \Val{V} + \unit }
		  \qquad
     \inference[\textsf{fold}]
			{\vdash^\delta M : \Val{M}}
			{\vdash^\delta M: \Val{M}}\qquad
		 \inference[\textsf{unfold}]
			{\vdash^{\delta}  M : \Val{M}}
			{\vdash^{\delta}  M: \Val{M} }
			
      \\
      \\
      \inference[\textsf{c}]
      {x:\ubang,y:\ubang\vdash^\delta M : \Val{M}}
      {z:\ubang\vdash^\delta \sub{M}{z/x,z/y} : \Val{M} + \unit}
      \qquad
      \inference[\textsf{w}]
      {\vdash^\delta M : \Val{M}}
      {x:\delta\vdash^\delta M : \Val{M} }
      \\\\
      \inference[\textsf{lam}]
      {\vdash^\delta M : \Val{M}}
      {\vdash^\delta \lambda x.M:  \Val{M}}
      \qquad
      \inference[\textsf{app}]
      {\vdash^\delta M_1 :\Val{M_1} & \vdash^\delta M_2:\Val{M_2}}
      {\vdash^\delta M_1M_2 : \Val{M_1}+\Val{M_2} + \Mthree}
      \\\\
      \inference[\textsf{$\S$-prom}]
	      {\vdash^{\delta} M : \Val{M}}
	      {\vdash^{\delta+1} \S M:  \S \Val{M} + \Mfour}
	      \qquad
      \inference[\textsf{$\S$-elim}]
	      {\vdash^{\delta} V : \Val{V}& \Gamma\vdash^\delta M:\Val{M}}
	      {\vdash^\delta \letp{x}{V}{M} : \Val{M}+ \Val{V}+\Mthree}
	    \\\\
      \inference[\textsf{$\oc$-prom}]
      {\vdash^{\delta} M : \Val{M}}
      {\vdash^{\delta+1} \oc M:  F(\Val{M})}
      \qquad
      \inference[\textsf{$\oc$-elim}]
      {\vdash^{\delta} V : \Val{V}& \Gamma\vdash^\delta M:\Val{M}}
      {\vdash^\delta \letb{x}{V}{M} : \Val{M}+ \Val{V}+ \Mthree}
   \end{array}
   $$
   
   \caption{Inferring a bound from a $\LALrefmu$ typing judgment}
   \label{tab-p}
 \end{table}

 We now prove the soundness of our model with respect to $\LALrefmu$
 and as a corollary the bounded-time termination theorem.

 In Table~\ref{tab-p} is described how to infer an element of $\M$
 from a $\LALrefmu$ typing judgment: the notation $\Val{M}$ corresponds
 to the element of $\M$ already inferred from the typing judgment
 of $\Val{M}$, and each rule corresponds to the way $\Val{M}$ is
 built.

 \begin{definition}
 We use the notations $\Vlist{V}$, $\Vlist{p}$ and $\Vlist{y}$ to denote 
 respectively a list of values $[V_1, \dots, V_n]$, a list $[p_1, \dots, p_n]$  
 of elements of $\M$ and a list of variables $[y_1,\dots, y_n]$. 
 If $M$ is a term, we denote by $M[\Vlist{V}/\Vlist{y}]$ the 
 term $M[V_1/y_1,\dots, V_n/y_n]$. 
 If $\Vlist{p}$ is a list of elements of $\M$ and $\dagger \in\{\oc,\S\}$, we denote
 by $\dagger \Vlist{p}$ the list $[\dagger p_1, \dots,\dagger p_n]$. 
 We also define $\sum \Vlist{p}$ to be the sum $\sum_{1\leq i \leq n} p_i$. 
 \end{definition}

 If $A$ is a type then we define $\lambda A$ as $A$ itself. 
 Suppose $\Gamma = x_1 : (e_1, A_1),\dots, x_n : (e_n, A_n)$. 
 Then the notation $(\Vlist{V}, \Vlist{p})\Vdash^\delta \Gamma$ stands for
 $(W_i, p_i) \in \preInter{e_i A_i}{\delta}$ for $1\leq i\leq n$
 with $W_i = V_i$ if $e_i = \lambda$ and $W_i = \dagger V_i$ if $e_i = \dagger$. 
 
 \begin{example}
   If we have $(\Vlist{V}, \Vlist{p}) \Vdash^\delta (x_1 : (\lambda, A_1), x_2 : (\S, A_2), x_3 : (\oc, A_3))$
   then $\Vlist{V} = [V_1, V_2, M_3]$ and $\Vlist{p} = [p_1, \S p_2, \oc p_3]$ such
   that $(V_1,p_1) \in \preInter{A_1}{\delta}$, 
   $(\S V_2,\S p_2) \in \preInter{\S A_2}{\delta}$ and $(\oc V_3,\oc p_3)\in \preInter{\oc A_3}{\delta}$.
 \end{example}
 
 \begin{theorem}[Adequacy]\label{th:adequacy}
   Suppose that $R;\Gamma \vdash^\delta M : C$. 
   Let $(\Vlist{V},\Vlist{p}) \Vdash^\delta \Gamma$, 
   Then, for any $\STa$ such that $\interpR{R}{\delta} \Stincl \STa$,
$$(\sub{M}{\Vlist{V}/\Vlist{x}},  \Val{M} +  \sum \Vlist{p}) \in \Inter{C}{\delta}{\STa}$$
   Moreover, if $M$ is a value, then we have $(\sub{M}{\Vlist{V}/\Vlist{x}},  \Val{M} +  \sum \Vlist{p}) \in \overline{\preInter{C}{\delta}}$.
  \end{theorem} 

 Before giving the proof of Theorem \ref{th:adequacy}, let's begin 
 by an essential remark about the definition of the interpretation
 of a region context $R$. 

   \begin{remark}\label{rmk:store-decompose}
  Let $\interpR{R}{\delta+1} \Stincl R'$. If $(S,s)\in R'$,
  it can be uniquely written $(\StL{S}{\delta} \uplus {S}^{=\delta+1} \uplus S_r, \S s_\delta + s_{\delta+1} + s_r)$
  where $(\StL{S}{\delta}, s_\delta)\in\interpR{R}{\delta}$, $({S}^{\delta+1}, s_{\delta+1})\in \interpR{R}{=\delta+1}$
  and $dom(S_r) \subseteq \setp{r_i}{\delta_i > \delta+1}$. 
  Hence, if we form the set $R'_\delta = \setp{(S,s)}{(S, \S s+s_{\delta+1} + s_r) \in R'\wedge 
  (\StL{S}{\delta},s)\in \interpR{R}{\delta}}$
  is such that $\interpR{R}{\delta} \Stincl R'_\delta$. 
  \end{remark}

  We also need two intermediate lemmas. The first one is 
  about the promotion rule for $\S$. 

  \begin{lemma}[Promotion]\label{lm:promotion}
  Suppose that for any $\STa$ such that $\interpR{R}{\delta}\Stincl \STa$,
  $(M,m)\in\Inter{A}{\delta}{\STa}$ holds.
  Then for any $\STa$ such that $\interpR{R}{\delta+1} \Stincl \STa$,
  we have $(\S M, \S m+ \Mfour)\in\Inter{\S A}{\delta+1}{\STa}$.
  \end{lemma}
  \begin{proof}
  Take $\STa$ such that $\interpR{R}{\delta+1}\Stincl \STa$,
  $(E,e)\in {(\Inter{\S A}{\delta+1}{\STa})}^{\bot_{\STa}}$ and
  $(S',s')\in \STa$. We have $s'= \S s_\delta + s_0$ with $(\StL{S'}{\delta}, s_\delta)\in\interpR{R}{\delta}$
  (see Remark \ref{rmk:store-decompose}). 
  We want to show that $(\conf{!M,E,S'}, e(\S m+\Mfour+s')) \in \Bot$.
  By anti-reduction and $\leq$-saturation (by monoidality of $\S$), it suffices to show  
  $(\conf{M,\S.E,S'}, e(\S (m + s_\delta)+s_0+\Mthree)) \in \Bot$.
  We pose $\STa' = \setp{(S,s)}{(S,\S s+s_0)\in  \STa\wedge (\StL{S}{\delta}, s)\in \interpR{R}{\delta}}$. 
  Since $\interpR{R}{\delta} \Stincl \STa'$, and $(S', s_\delta) \in \STa'$, it is sufficient to prove
    $(\S.E, \lambda x.e(\S x+\Mthree +s_0)) \in {(\Inter{A}{\delta}{\STa'})}^{\bot_{\STa'}} = \preInter{A}{\delta}^{\bot_{\STa'}}$.
  So let $(V,v) \in \preInter{A}{\delta}$ and $(S,s)\in \STa'$. We know that 
  $(\S V,\S v) \in \preInter{\S A}{\delta+1}$. Moreover, $(S,\S s+s_0)\in \STa$.
  Since $(E,e)\in\preInter{\S A}{\delta+1}^{\bot_{\STa}}$,
  we have $(\conf{\S V,E,S}, e(\S v+\S s+s_0))\in\Bot$.
  So by $\reduc$-saturation, $(\conf{V, \S.E, S}, e(\S v+\unit+\S s+s_0))\in\Bot$.
  Hence, by distributivity of $\S$, we obtain
  $(\conf{V, \S.e, S}, e(\S (v+s) + \Mtwo+ \unit + s_0)) \in \Bot$. 
  \end{proof}
	
	The proof of this last lemma is very important, since it justifies many 
  design choices of our model.
	\begin{itemize}
	  \item Its proof crucially relies on the fact that in the definition 
	of the region context interpretation $\interpR{R}{\delta}$, 
	each value is guarded by a modality $\S$ and not by a modality
	$!$. Indeed, it requires the monoidality property, which is 
	true for $\S$ but not for $\oc$:
	$\forall p,q\in\M, \S(p+q) \leq \S p+ \S q$.
	  \item It also relies on the fact that we can consider 
	  any set of store $\STa$ such that $\interpR{R}{\delta} \Stincl \STa$,
	  which is also built-in in our interpretation of the 
	  linear arrow $\multimap$. 
  \item It also justifies the fact that we need to consider $\M$-contexts,
  which are not only of the form $x\mapsto x + p$, since in this proof we use 
  functions of the form $x\mapsto \S x + p$. 
	\end{itemize}

  We also prove a commutation  lemma, which is needed when one wants
  to use biorthogonality in presence of call-by-value and modalities like $!$
  or $\S$. 

  \begin{lemma}[Commutation]\label{lm:commutation}
  Suppose that
  $\forall (V,p) \in X, (M[V/x], q+p) \in Y$.
  Then for every $\STa$, we have 
  $\forall (V,p) \in X^{\bot_\STa\bot_\STa}, (M[V/x], q+p+\Mtwo)\in Y^{\bot_\STa\bot_\STa}$.
  \end{lemma}
  \begin{proof}
  The proof mainly uses the property of $\reduc$-saturation 
  and of $\reduc$-closure. Here are a series of implications.
  \begin{eqnarray*}
   & & \forall (V,p) \in X, (M[V/x], q+p) \in Y\\
   & \Rightarrow & \forall (V,p)\in X, \forall (E,e)\in Y^{\bot_\STa}, \forall (S,s)\in \STa,
   (\conf{M[V/x], E, S}, e(q+p+s))\in\Bot\\
   & \Rightarrow & \forall (V,p)\in X, \forall (E,e)\in Y^{\bot_\STa}, \forall (S,s)\in \STa,\\
   & & 
   (\conf{V, (\lambda x.M)\odot E, S}, e(q+p+s+\Mtwo))\in\Bot\\
   & \Rightarrow & \forall (E,e)\in Y^{\bot_\STa}, ((\lambda x.M)\odot E, \lambda x.e(q+x+s+\Mtwo))\in X^{\bot_\STa}\\
   & \Rightarrow & \forall (E,e)\in Y^{\bot_\STa}, ((\lambda x.M)\odot E, \lambda x.e(q+x+s+\Mtwo))\in X^{\bot_\STa\bot_\STa\bot_\STa}\\
   & \Rightarrow & \forall (V,p)\in X^{\bot_\STa\bot_\STa}, \forall (E,e)\in Y^{\bot_\STa}, \forall (S,s)\in \STa,
   \\
    & & (\conf{V, (\lambda x.M)\odot E, S}, e(q+p+s+\Mtwo))\in\Bot\\
   & \Rightarrow & \forall (V,p)\in X^{\bot_\STa\bot_\STa}, \forall (E,e)\in Y^{\bot_\STa}, \forall (S,s)\in \STa,\\
  & &  (\conf{M[V/x], E, S}, e(q+p+s+\Mtwo))\in\Bot\\
   & \Rightarrow & \forall (V,p)\in X^{\bot_\STa\bot_\STa}, (M[V/x],q+p+\Mtwo)\in Y^{\bot_\STa\bot_\STa}\\
   \end{eqnarray*}
  \end{proof}

  We can now prove Theorem \ref{th:adequacy}.

   \begin{proof} 
   This theorem is proved by induction on the typing judgment.
   When we consider a value, we only prove the second statement, since it implies the 
   first (by the properties of biorthogonality). 
   \begin{itemize}    
   \item[$\mathsf{(v)}$] This case is immediate by substitution.  
   \item[$\mathsf{(r)}$, $\mathsf{(u)}$, $\mathsf{(int)}$] and $\mathsf{(arith)}$ These two cases are trivial, by definition
   of $\preInter{\tertype}{\delta}$ and $\preInter{\regtype{r}{A}}{\delta}$. 
   \item[$\mathsf{(w)}$] This case is just an application of $\leq$-saturation. 
   \item[$\mathsf{(fold)}$, $\mathsf{(unfold)}$] The two fixpoint rules are easy, since
   we have $\preInter{\mu X.A}{\delta} = \preInter{A[\mu X.A/X]}{\delta}$.
	\item[$\mathsf{(lam)}$] If the last rule used is the introduction of $\lambda$:		
	$$\inference
		{R;\Gamma,\hyp{y}{\uaff}{A} \vdash^{\delta} N:B}
		{R;\Gamma \vdash^{\delta} \lambda y.N:A \multimap B}$$
	     We take $(\Vlist{V}, \Vlist{p}) \Vdash^\delta \Gamma$. 
		We denote by $N' = N[\Vlist{V}/\Vlist{x}]$ and
			$p' = \sum \Vlist{p}$. 			
	By induction hypothesis we know
	that for every $\STa$ such that $\interpR{R}{\delta}\Stincl \STa$ and 
	every $(V',v')\in\preInter{A}{\delta}$, 
		$(\sub{N'}{V'/x}, \Val{N'} + p' + v') \in \Inter{B}{\delta}{\STa}$.
	That means exactly that $(\lambda x.N', \Val{\lambda x.N} + p') \in \preInter{A \multimap B}{\delta}$ 
	where $\Val{\lambda x.N} = \Val{N}$.

\item[$\mathsf{(app)}$] Here, to simplify the presentation, we suppose
the contexts are empty (it does not change the argument).
	$$\inference 
	  {R;\vdash^\delta M:A \multimap B & R;\vdash^\delta N:A}
	  {R;\vdash^\delta MN:B}$$
Take $\STa$ such that $\interpR{R}{\delta}\Stincl \STa$. 
By induction hypothesis we have 

	$$
		\begin{array}{l}
		  (M, \Val{M}) \in \Inter{A\multimap B}{\delta}{\STa}\\
      (N,\Val{N})  \in \Inter{A}{\delta}{\STa}
    \end{array}
  $$
Take $(E,e) \in {(\Inter{B}{\delta}{\STa})}^{\bot_{\STa}}$ and $(S,s)\in \STa$. 
We want to show that 
$$(\conf{MN,E,S},e(\Val{MN}+s)) \in \Bot$$
where $\Val{MN} = \Val{M} + \Val{N} + \Mthree$. But 
$\conf{MN,E,S} \reduc \conf{N,M\odot E,S}$. Then, it suffices
to show  
  $(M \odot E,\lambda x. e(\Val{M}+ \Mtwo+x)) \in {\preInter{A}{\delta}}^{\bot_{\STa}}$.
Take $(V_A,v_A) \in \preInter{A}{\delta}$ and $(S',s')\in \STa$. Now we have to prove 
$$(V_A,M\odot E,S' ,e(\Val{M} + v_A + \Mtwo+s')) \in \Bot$$
But $\conf{V_A,M\odot E,S'} \reduc \conf{M,V_A \cdot E,S'}$, so
by $\reduc$-saturation we only have to prove
$(V_A \cdot E,S',\lambda x.e(x+v_A + \unit))\in \preInter{A\multimap B}{\delta}^{\bot_{\STa}}$.
Let $(\lambda x.P,p) \in \preInter{A\multimap B}{\delta}$ and $(S,s)\in \interpR{R}{\delta}$. 
We have $\conf{\lambda x.P,V_A\cdot E,S} \reduc \conf{\sub{P}{V_A/x},E,S}$.
But since $(\sub{P}{V_A/x}, p+v_A) \in \Inter{B}{\delta}{\STa}$, we have 
$$(\conf{\sub{P}{V_A/x},E,S}, e(p+v_A+s)) \in \Bot$$
Hence, we conclude that $(V_A \cdot E, \lambda x.e(x+v_A + \unit)) \in
\preInter{A \multimap B}{\delta}^{\bot_{\STa}}$ by $\reduc$-saturation.

\item[$\mathsf{(st)}$] In this case, we have the following typing rule
$$
\inference
{r:(\delta,\S C) \in R \\ R;\Gamma \vdash^{\delta-1}
 V : C}
{R;\Gamma \vdash^\delta \st{r}{\S V}:\tertype} 
$$
Here, it is safe to consider that $\Gamma = \emptyset$
since $V$ is closed (the case $\Gamma\neq\emptyset$ is recovered
from the case $\Gamma = \emptyset$ by $\leq$-saturation). Let $\STa$ such that 
$\interpR{R}{\delta} \Stincl \STa$, $(E, e) \in \preInter{\tertype}{\delta}^{\bot_{\STa}}$
and $(S,s)\in \STa$. We want to prove
  $$(\conf{\st{r}{\S V}, E, S}, e(\S \Val{V} + \unit + s))\in\Bot$$
But we have $(S, s+ \S \Val{V}) \in \STa$ so 
  $$(\conf{*, E, S\uplus\{r\rightarrow \S V\}}, e(s+\S \Val{V}) \in \Bot$$
Hence by $\reduc$-saturation and monotonicity of $e$, we obtain the result.

\item[$\mathsf{(get)}$] We now justify the $get$ typing rule, which says that
if $(r : (\delta, \S A) \in R$, then 
$$R;\vdash^\delta \get{r}:\S A$$
Let $\STa$ be such that $\interpR{R}{\delta}\Stincl \STa$,
$(E,e) \in \preInter{\S A}{\delta}^{\bot_{\STa}}$ and $(S,s) \in \STa$. 
So consider one possible decomposition $S = S' \uplus \{ r \rightarrow \S V \}$,
then we can decompose $s = s' + \S s_V$ with $(V, s_V) \in \Inter{A}{\delta-1}{\STa'}$ 
for any $\STa'$ such that $\interpR{R}{\delta-1}\Stincl \STa'$.
We then have
  $$\conf{get(r), E, S} \reduc \conf{\S V, E, S'}$$
Moreover, by Lemma \ref{lm:promotion}, $(\S V, \S s_V + \Mfour)\in\Inter{\oc A}{\delta}{\STa}$ and
$(S', s') \in \STa$. Hence,   
$$(\conf{\S V, E, S'}, e(\S s_V + \Mfour + s')) \in \Bot$$ 
Hence, by $\reduc$-saturation $(\conf{get(r), E, S}, e(\Mfive + s))$.
We conclude that $(get(r),\Mfive) \in \Inter{\S A}{\delta}{\STa}$.

   \item[$\mathsf{(c)}$] We want to justify the contraction rule
      $$\inference
					{R;\Gamma,\hyp{y}{\ubang}{A},\hyp{z}{\ubang}{A}\vdash^\delta M:B}
					{R;\Gamma,\hyp{y}{\ubang}{A}\vdash^\delta \sub{M}{y/z}:B}$$			
	     We take $(\Vlist{W},\Vlist{p}) \Vdash^\delta \Gamma$. 
		We denote by $M' = M[\Vlist{W}/\Vlist{x}]$ and
			$p' = \sum \Vlist{p} $. Let $\STa$ such that $\interpR{R}{\delta}\Stincl \STa$.
		We take $(V,v) \in \preInter{A}{\delta - 1}$.
		By induction hypothesis we have 
  		$(\sub{M'}{V/y,V/z},\Val{M} + p' + \oc v + \oc v) \in \Inter{B}{\delta}{\STa}$.
		Since $\norm{\oc p+\oc p} \leq \norm{\oc p +\Mtwo}$ for any $p\in\M$,
		we conclude that 
 			$(\sub{\sub{M'}{y/z}}{V/y},\Val{M} +p'+ \oc v) \in \Inter{R}{B}{\delta}{\STa}$.
 	  The case where $M'$ is a value is similar. 

\item[$\mathsf{(}!-\mathsf{prom)}$]  The $!$ promotion rule is as follows.
       $$\inference
					{R;\hyp{x}{\uaff}{A} \vdash^{\delta-1} V : B}
					{R;\hyp{x}{\ubang}{A}\vdash^{\delta}\oc V :\oc B}$$		 
		Let $(V',p)\in \preInter{A}{\delta-1}$. 
		We know by induction that there is some
		$q$ such that $q \leq \Val{V}+p$ with
	  $(V[V'/x], q)\in\preInter{B}{\delta-1}$. 
	  We then have immediately that 
	  $$(\oc V[V'/x], \oc q)\in \preInter{\oc B}{\delta}$$
	  and so 
	  $$(\oc V[V'/x], \oc (\Val{V}+p))\in\overline{\preInter{\oc B}{\delta}}$$
	  But, since $ \norm{\oc (\Val{V}+p)} \leq\norm{F(\Val{V}) + \oc p}$
		by Property \ref{prop:functoriality}. Hence, by $\leq$-saturation,
		$( (\oc V)[V'/x], F(\Val{M}) + \oc p)\in\overline{\preInter{\oc B}{\delta}}$.

\item[$\mathsf{(}\S-\mathsf{prom)}$]   Here, suppose the the rule is written
       $$\inference
					{R; x_1 : (\uaff, A_1),\dots, x_n : (\uaff, A_n) \vdash^{\delta-1} M : C}
					{R;x_1 : (\S, A_1),\dots, x_n : (\S, A_n) \vdash^\delta \S M : \S C}$$
					
Let's take $(V_i,p_i) \in \preInter{A_i}{\delta-1}$ for $1\leq i \leq n$
and $(W_j,q_j)\in\preInter{B_j}{\delta-1}$ for $1\leq j\leq k$.
We know by induction hypothesis that for any $\STa$ such that $\interpR{R}{\delta-1}
\Stincl \STa$, we have $(M[V_i/x_i,W_j/y_j],\Val{M} + \sum_i p_i+\sum_j q_j)\in\Inter{C}{\delta-1}{\STa}$. 
Take $\STa'$ such that $\interpR{R}{\delta}\Stincl \STa'$. 
Hence, by Lemma \ref{lm:promotion}, 
$$(\S M[V_i/x_i,W_j/y_j], \S (\Val{M}+\sum_i p_i+\sum_j q_j) + \Mfour) \in \Inter{\S C}{\delta}{\STa'}$$
But, by Property \ref{prop:monoid}, we have
$\norm{\S (\Val{M}+\sum_i p_i+\sum_j q_j )} \leq \norm{\S \Val{M} + \sum_i \S p_i +\sum_j \oc q_j}$.
We can conclude 
$$(\S M[V_i/x_i,W_j/y_j], \Val{\S M}+\sum_i \S p_i + \sum_j \oc q_j) \in \Inter{\S C}{\delta}{\STa'}$$

\item[$\mathsf{(}\S-\mathsf{elim)}$]
The rule is as follows (where the variables respectively associated
to the contexts $\Gamma_\lambda, \Delta_\lambda, \Gamma_!,\Delta_!, \Gamma_\S$
and $\Delta_\S$ are noted $\Vlist{x},\Vlist{x'},\Vlist{y},\Vlist{y'},\Vlist{z},
\Vlist{z'}$). 
  $$
  \inference
   {R;\Gamma  \vdash^{\delta} V: \dagger B& 
    R;\Delta,\hyp{x_\S}{\S}{B} \vdash^\delta M:C}
  {R;\Gamma,\Delta \vdash^\delta \letp{x_\S}{V}{M}:C}
  $$
	     We take $(\Vlist{V_\Gamma},\Vlist{p}) \Vdash^\delta \Gamma$ and
	     $(\Vlist{V_\Delta},\Vlist{p'}) \Vdash^\delta \Delta$. 
  We pose
	$$
		\begin{array}{rcl}
		  V' & =&  V[\Vlist{V_\Gamma}/\Vlist{x}]\\
      p_V & = & \sum \Vlist{p}\\
      M' &= & M[\Vlist{V_\Delta}/\Vlist{x'}]\\
      p_M & = & \sum \Vlist{p'}
    \end{array}
  $$
 Let's take $\STa$ such that $\interpR{R}{\delta}\Stincl \STa$.
 Then, by induction hypothesis, we know 
  that if $(\S W, \S p_W) \in \preInter{\S B}{\delta}$, 
  $$(M'[W/x_\S], \Val{M} + p_M + \S p_W) \in \Inter{C}{\delta}{\STa}$$
  By $\reduc$-saturation (by considering a context), 
  we obtain that for any $(\S W, \S p_W) \in \preInter{\S B}{\delta}$
  $$(\letp{x_\S}{\S W}{M'}, \Val{M} + \unit + p_M + \S p_W)\in\Inter{C}{\delta}{\STa}$$
  But we also know by induction hypothesis that $(V', \Val{V}+p_V)\in \Inter{\S B}{\delta}{\STa}$.
  Hence, by Lemma \ref{lm:commutation}, since $\Inter{C}{\delta}{\STa} 
   = \preInter{\S C}{\delta}^{\bot_\STa \bot_\STa}$,    
   we obtain
   $$(\letp{x_\S}{V'}{M'}, \Val{M} + \Val{V} + \Mthree + p_V + p_M) \in \Inter{C}{\delta}{\STa}$$
   
\item[$\mathsf{(}!-\mathsf{elim)}$]  This case is completely similar 
to the previous case, except we have to replace every mention 
of $\S$ by $\oc$. 
   \end{itemize}
 \end{proof}
 
	As a corollary of the adequacy theorem, we obtain the announced bounded-time termination theorem
	for $\LALrefmu$ programs. 
 \begin{proof}[Termination theorem (Theorem \ref{thm:termination})]
   This theorem is proved using adequacy together with 
   Property \ref{prop-redcan}. Indeed, 
   we know that $\conf{M,\empStack,\empStore}$ terminates
   in at most $\norm{\Val{M}}$ steps. But it is easy to see
   that only the promotion rules for $\S$ and $\oc$ make
   the value of $\norm{Val{M}}$ increases significantly: the
   degree of the third component of $\Val{M}$ (which is
   a polynomial) is bounded by a function of the depth
   of $M$. A similar argument is made more precise in \cite{dalagobll},
   for instance.
 \end{proof}


\section{Related Work}
\label{sec:related}
 \partitle{Approximation modality}	In a series of two papers, Nakano introduced
a normalizing intuitionistic type system that features recursive types, which
are guarded by a modality $\bullet$ (the approximation modality). 
Nakano also defines an indexed realizability semantics for this
type system. The modality $\S$ plays
in our work almost the same role as $\bullet$: it makes the
index increase. We claim that when we forget the quantitative part 
of our model, we obtain a model for a language with guarded references, that
can be extended to handle control operators, based
on a fragment of Nakano's type system: the only difference
is that the $\bullet$ modality does not enjoy digging anymore
(in presence of control operators, this principle would 
break normalization). \newline

\partitle{Stratified semantics for light logics} Several semantics for 
the "light" logics have been proposed, beginning with fibered phase models \cite{kanovich2003phase},
a truth-value semantics for \class{LLL}. We can also mention 
stratified coherent spaces \cite{baillot2004stratified}. 
These two models are indexed, like ours, but while the indexing 
is used to achieve completeness with respect to the logic, we use it
to interpret fixpoints and references. \newline

\partitle{Reactive programming} In \cite{krishnaswami2012higher}, Krishnaswami \& al. have proposed a 
type system for a discrete-time reactive programming 
language that bounds the size of the data flow graph produced
by programs. It is based on linear types and a Nakano-style
approximation modality, thus bounding space consumption
and allowing recursive definitions at the same time. They
provide a denotational semantics based on both ultrametric
semantics and length spaces. These latter, 
introduced by Hofmann \cite{hofmann2003linear} constitute
the starting point of the quantitative realizability presented
here.

\section{Research directions}
\label{sec:conclusion}
We see several possible directions we plan to 
explore.\newline

\partitle{Control operators} Since we use a biorthogonality-based model,
it is natural to extend the language with control operators. Adding
the call-cc operator can be done, but it requires
to add a modality type $?$ for \emph{duplicable contexts}. This involves
some technical subtleties in the quantitative part, like the symmetrization
of the notion of $\M$-contexts. Indeed, in our framework, a $\M$-context can be
used to \emph{promote} a weight associated to a term, but with this new $?$
type, a weight associated to a term would need to be able to promote a weight
associated to a stack. \newline

\partitle{Multithreading} In the original work of Amadio and Madet \cite{AmadioMadet2011},
the language features regions but also multithreading. It is possible to 
add it to $\LALrefmu$ but so far, it seems difficult to adapt the quantitative
framework for this extension. It may be possible to adapt the notion
of \emph{saturated store} presented in \cite{amadio2009stratified}, but with
a boundedness requirement on it. We plan to explore this direction in the future.

\bibliographystyle{plain}
\bibliography{biblio}

\begin{thebibliography}{10}

\bibitem{amadio2009stratified}
R.~Amadio.
\newblock On stratified regions.
\newblock {\em Programming Languages and Systems}, pages 210--225, 2009.

\bibitem{AppelStepIndexed2001}
A.W. Appel and D.~McAllester.
\newblock An indexed model of recursive types for foundational proof-carrying
  code.
\newblock {\em ACM Transactions on Programming Languages and Systems (TOPLAS)},
  23(5):657--683, 2001.

\bibitem{asperti1998light}
A.~Asperti.
\newblock {Light affine logic}.
\newblock In {\em Thirteenth Annual IEEE Symposium on Logic in Computer
  Science, 1998. Proceedings}, pages 300--308, 1998.

\bibitem{baillot2004stratified}
P.~Baillot.
\newblock Stratified coherence spaces: a denotational semantics for light
  linear logic.
\newblock {\em Theoretical Computer Science}, 318(1):29--55, 2004.

\bibitem{brunel2012quantitative}
A.~Brunel.
\newblock Quantitative classical realizability.
\newblock {\em submitted}, 2012.

\bibitem{DalLago2010Process}
U.~Dal~Lago, S.~Martini, and D.~Sangiorgi.
\newblock Light logics and higher-order processes.
\newblock {\em Electronic Proceedings in Theoretical Computer Science}, 41,
  2010.

\bibitem{girard1995light}
J-Y. Girard.
\newblock {Light linear logic}.
\newblock In {\em Logic and Computational Complexity}, pages 145--176.
  Springer, 1995.

\bibitem{hofmann2003linear}
M.~Hofmann.
\newblock Linear types and non-size-increasing polynomial time computation.
\newblock {\em Information and Computation}, 183(1):57--85, 2003.

\bibitem{krivine-realizability}
J-L. Krivine.
\newblock {Realizability in classical logic. Course notes of a series of
  lectures given in the University of Marseille, may 2004 (last revision: july
  2005)}.
\newblock {\em Panoramas et syntheses, Soci{\'e}t{\'e} Math{\'e}matique de
  France}, 2005.

\bibitem{lafont2004soft}
Y.~Lafont.
\newblock {Soft linear logic and polynomial time}.
\newblock {\em Theoretical Computer Science}, 318(1-2):163--180, 2004.

\bibitem{dalagobll}
U.~Dal Lago and M.~Hofmann.
\newblock Bounded linear logic, revisited.
\newblock In Pierre-Louis Curien, editor, {\em Typed Lambda Calculi and
  Applications}, volume 5608 of {\em Lecture Notes in Computer Science}, pages
  80--94. 2009.

\bibitem{dal2010semantic}
U.~Dal Lago and M.~Hofmann.
\newblock A semantic proof of polytime soundness of light affine logic.
\newblock {\em Theory of Computing Systems}, 46:673--689, 2010.

\bibitem{DalLago20112029}
U.~Dal Lago and M.~Hofmann.
\newblock Realizability models and implicit complexity.
\newblock {\em Theoretical Computer Science}, 412(20):2029 -- 2047, 2011.
\newblock Girard's Festschrift.

\bibitem{Lucassen88}
J.~M. Lucassen and D.~K. Gifford.
\newblock Polymorphic effect systems.
\newblock In {\em Proceedings of the 15th ACM SIGPLAN-SIGACT symposium on
  Principles of programming languages}, POPL '88, pages 47--57, New York, NY,
  USA, 1988. ACM.

\bibitem{AmadioMadet2011}
A.~Madet and R.~Amadio.
\newblock An elementary affine $\lambda$-calculus with multithreading and side
  effects.
\newblock {\em Typed Lambda Calculi and Applications}, pages 138--152, 2011.

\bibitem{kanovich2003phase}
M.~Okada M.I.~Kanovich and A.~Scedrov.
\newblock Phase semantics for light linear logic.
\newblock {\em Theoretical Computer Science}, 294(3):525--549, 2003.

\bibitem{nakano2000modality}
H.~Nakano.
\newblock A modality for recursion.
\newblock In {\em Logic in Computer Science, 2000. Proceedings. 15th Annual
  IEEE Symposium on}, pages 255--266. IEEE, 2000.

\bibitem{krishnaswami2012higher}
N.~Benton N.R~Krishnaswami and J.~Hoffmann.
\newblock Higher-order functional reactive programming in bounded space.
\newblock In {\em Proceedings of the 39th annual ACM SIGPLAN-SIGACT symposium
  on Principles of programming languages}, pages 45--58. ACM, 2012.

\bibitem{baillot2010polytime}
M.~Gaboardi P.~Baillot and V.~Mogbil.
\newblock A polytime functional language from light linear logic.
\newblock {\em Programming Languages and Systems}, pages 104--124, 2010.

\bibitem{Terui07}
Kazushige Terui.
\newblock Light affine lambda calculus and polynomial time strong
  normalization.
\newblock {\em Archive for Mathematical Logic}, 46(3-4):253--280, 2007.

\end{thebibliography}

\end{document}